\documentclass[twocolumn]{autart}    % Enable this line and disable the  preceding line to obtain a two-column  document whose style resembles the printed Automatica style.

\usepackage{graphicx}          

\usepackage{cite}
\usepackage{amsmath,amssymb,amsfonts}
\usepackage{epstopdf}

\usepackage{algorithm}
\usepackage{etoolbox}

\let\classAND\AND
\let\AND\relax
\usepackage{algorithmic}
\let\AND\classAND
\AtBeginEnvironment{algorithmic}{\let\AND\algoAND}

\begin{document}

\begin{frontmatter}
%\runtitle{Insert a suggested running title}  % Running title for regular  papers but only if the title   is over 5 words. Running title is not shown in output.

\title{Adaptive Tube MPC: Beyond a Common Quadratically Stabilizing Feedback Gain\thanksref{footnoteinfo}} % Title, preferably not more  than 10 words.

\thanks[footnoteinfo]{Corresponding author A.~Dey.}

\author{Anchita Dey}\ead{anchitadey.ee.india@gmail.com},    % Add the 
\author{Shubhendu Bhasin}\ead{sbhasin@ee.iitd.ac.in}               % e-mail address 

\address{Electrical Engineering Department, Indian Institute of Technology Delhi, Hauz Khas, New Delhi, Delhi 110016, India}  % Please supply % full addresses % here.

                             % keyword list or with the  help of the Automatica  keyword wizard

\begin{abstract}                          % Abstract of not more than 200 words.
\iffalse
In this work, we develop an adaptive tube framework for model predictive control (MPC) of discrete-time linear time-invariant systems subject to both parametric uncertainty and additive disturbances. Unlike conventional tube MPC methods that rely {\emph{on fixed tube geometry, tightened constraints and terminal set,}} on a fixed tube geometry, tightened constraints, and a terminal set, the proposed approach incorporates online parameter learning to refine the parametric uncertainty set, update parameter point estimates, and accordingly adapt the components of the constrained optimal control problem. The refined model, used for prediction, leads to redesigned feedback gains, terminal ingredients, and tube cross-sections as uncertainty decreases, thereby reducing conservatism while preserving recursive feasibility, stability and constraint satisfaction. Further, the proposed framework does not hinge on the existence of a common quadratically stabilizing linear feedback gain for the entire parametric uncertainty set, which is a common and yet restrictive assumption in standard tube-based designs.
%
%
%
%
\fi
This paper proposes an adaptive tube framework for model predictive control (MPC) of discrete-time linear time-invariant systems subject to parametric uncertainty and additive disturbances. In contrast to conventional tube-based MPC schemes that employ fixed tube geometry and constraint tightening designed for worst-case uncertainty, the proposed approach incorporates online parameter learning to progressively refine the parametric uncertainty set and update the parameter estimates. These updates are used to adapt the components of the MPC optimization problem, including the prediction model, feedback gain, terminal set, and tube cross-sections. As the uncertainty set contracts, the required amount of constraint tightening reduces and the tube shrinks accordingly, yielding less conservative control actions. Recursive feasibility, robust constraint satisfaction, and closed-loop stability are formally established. Furthermore, the framework does not require the existence of a common quadratically stabilizing linear feedback gain for the entire parametric uncertainty set, thereby relaxing a standard assumption in existing tube-based MPC formulations. Numerical examples illustrate the effectiveness of the proposed approach.

%197 words
%The prediction model, informed by the parametric uncertainty set update, leads to  
%Tube adaptation 
%In addition, since a single nominal system is used at each time for the COCP predictions, the terminal set is constructed using the time-varying nominal parameters and LQR feedback gain. As a result, the proposed theory does not hinge on the existence of a quadratically stabilizing feedback gain for the parametric uncertainty set. 
% Recursive feasibility, stability, and convergence for the proposed design are guaranteed, and a numerical example is provided for validation.
\end{abstract}

\end{frontmatter}

\section{Introduction}
Model predictive control (MPC) is an optimization-based control strategy that uses a system model to predict future system behavior and computes an optimal sequence of control inputs by minimizing a cost function of predicted states and inputs, subject to hard constraints. In the standard MPC formulation, the constrained optimal control problem (COCP) relies on an accurate model of the system dynamics for state prediction \cite{kouvaritakis2016model}. In the presence of modeling errors or disturbances, however, the true system evolution may deviate from the predicted trajectories, potentially leading to constraint violations and loss of recursive feasibility or stability. Since exact model knowledge is rarely available in practice, the problem of designing MPC schemes for systems affected by parametric uncertainty and additive disturbances has received considerable attention in the literature.

Robust MPC methods are mostly based on either min-max optimization or tube-based approaches and have been extensively studied in the literature. These methods are typically designed for systems with known parameters subject to unknown additive disturbances \cite{CHISCI20011019,langson2004robust,rakovic2012homothetic}, or for systems with parametric uncertainty, with or without additive disturbances \cite{fleming2015,KOTHARE19961361,buj2021}. In the latter case, however, the uncertain parameters are typically treated as unknown but fixed, and no mechanism is employed to learn or refine their estimates online. To reduce conservatism and improve performance, adaptive MPC schemes have recently been proposed \cite{jafari2017adaptive,dhar2021indirect,kohlerrampc,lorenzen2019robust,lu2023robust}, where the system parameters, the sets containing them, or both are updated using online data, while prediction errors are handled through robust tubes. 

Among robust MPC formulations, homothetic tube MPC \cite{langson2004robust,rakovic2012homothetic} provides a convenient parametrized structure for bounding prediction errors and formulating the COCP. Consequently, recent research has focused on reducing the conservatism associated with fixed tube constructions. For instance, \cite{luflexitube} proposes a flexible tube structure based on zonotopes, where the parameters defining the tube cross-sections are optimized online. The framework in \cite{kogelblending} considers systems with additive disturbances only and constructs tubes through an online-optimized blending of offline-computed invariant sets corresponding to different stabilizing feedback gains. Self-tuning tubes are proposed in \cite{tranosrusso} for systems with both parametric uncertainty and additive disturbances.

\iffalse
For systems with parametric uncertainties and constraints imposed on all states, most of the existing literature \cite{lorenzen2019robust,lu2023robust,anchita2025,anch,buj2021,jafari2017adaptive,dhar2021indirect,tranosrusso,kohlerrampc,luflexitube} assumes the existence of a quadratically stabilizing linear state feedback gain for the parametric uncertainty set, i.e., there exists a common feedback gain that can stabilize the origin of a linear time-invariant (LTI) system whose state and input matrices are formed from any element in the uncertainty set. However, this assumption substantially restricts the class of admissible uncertainty sets \cite{khargonekar1990robust,Gahinet1994,barmish1985necessary}. A general parametric uncertainty set may fail to be quadratically stabilizable, and even if quadratic stability can be established, a common linear feedback gain that stabilizes every element of the set may not exist \cite{Petersen1985}.
% One possible way out is to handle the parametric uncertainty as additive disturbance \cite[Sec.~5]{langson2004robust} with conservative upper bounds that may possibly lead to no feasible region for the COCP. 
The data-driven frameworks in \cite{berberich2020data,coulson2019data} are efficient alternative approaches that do not require a common feedback gain; however, owing to the offline formation of the Hankel matrix, their performance may degrade if there are sudden small changes in the system dynamics during online operation.
\fi

For systems with parametric uncertainty and constraints imposed on all states, most existing approaches \cite{lorenzen2019robust,lu2023robust,anchita2025,anch,buj2021,jafari2017adaptive,dhar2021indirect,tranosrusso,kohlerrampc,luflexitube} assume the existence of a common quadratically stabilizing linear state feedback gain for the entire uncertainty set \cite{khargonekar1990robust}. Specifically, it is assumed that there exists a common linear feedback gain that stabilizes the origin of a linear time-invariant (LTI) system whose state and input matrices correspond to any element of the parametric uncertainty set. This assumption, however, significantly restricts the class of admissible uncertainty sets \cite{Gahinet1994,barmish1985necessary}. In general, a parametric uncertainty set may not admit quadratic stabilization, and even when quadratic stability can be established, a linear feedback gain that stabilizes every element of the set may not exist \cite{Petersen1985}. Even in cases where such a common stabilizing gain exists, the construction of a common terminal set satisfying the state and input constraints for all admissible systems may still be infeasible or highly conservative, further restricting the practical applicability of such approaches.

% For systems with parametric uncertainty and constraints imposed on all states, most existing approaches \cite{lorenzen2019robust,lu2023robust,anchita2025,anch,buj2021,jafari2017adaptive,dhar2021indirect,tranosrusso,kohlerrampc,luflexitube} assume the existence of a common quadratically stabilizing linear state-feedback gain for the entire uncertainty set \cite{khargonekar1990robust}. Specifically, it is assumed that there exists a common feedback gain that stabilizes the origin of a linear time-invariant (LTI) system whose state and input matrices correspond to any element of the parametric uncertainty set. This assumption, however, significantly restricts the class of admissible uncertainty sets \cite{Gahinet1994,barmish1985necessary}. In general, a parametric uncertainty set may not admit quadratic stabilization, and even when quadratic stability can be established, a linear feedback gain that stabilizes every element of the set may not exist \cite{Petersen1985}.

Data-driven predictive control frameworks \cite{berberich2020data,coulson2019data} provide an alternative that avoids the need for a common stabilizing feedback gain. However, since these methods rely on an offline-constructed Hankel matrix derived from previously collected data, their performance may degrade if the system dynamics changes during online operation.

In this work, we propose the notion of \emph{adaptive tubes} for MPC of discrete-time LTI systems subject to both parametric uncertainty and additive disturbances. Specifically, we construct homothetic tubes whose cross-sectional geometry is updated online based on the evolving parameter point estimate and the refined parametric uncertainty set. The resulting adaptive tube framework avoids the common but restrictive assumption of the existence of a single quadratically stabilizing linear feedback gain for the entire uncertainty set \cite{lorenzen2019robust,lu2023robust,anchita2025,anch,buj2021,jafari2017adaptive,dhar2021indirect,tranosrusso,kohlerrampc,luflexitube}, thereby reducing conservatism.

To achieve this, the system matrices are decomposed into an estimated component and an unknown component, with the latter treated as part of a lumped additive disturbance. This representation enables the formulation of a tube-based COCP with a structure similar to that of standard homothetic tube MPC \cite{langson2004robust,rakovic2012homothetic}, while yielding improved constraint tightening through the use of separate forward reachable disturbance sets \cite{CHISCI20011019} rather than a single worst-case disturbance bound. The available state and input data are used to update both the parametric uncertainty set \cite{boydset} and the parameter point estimate through a projection-modified normalized gradient descent law \cite{ioannou2006adaptive}. The resulting reduction in parametric uncertainty leads to online recomputation of the tightened constraints, the tube cross-sectional shape, and the terminal set. Furthermore, the terminal ingredients are constructed using a feedback gain associated with the current parameter estimate rather than requiring a common gain for all elements of the uncertainty set. Although this adaptation modifies the underlying COCP and may affect feasibility, recursive feasibility is ensured by a backup formulation invoked when necessary. The proposed scheme also guarantees robust exponential stability \cite[Def.~2]{mayne2006robust} and boundedness of all signals.

The main contributions of this work are twofold. First, we introduce the notion of adaptive tubes for MPC of discrete-time LTI systems with parametric uncertainty and additive disturbances, where the tube cross-sections, terminal set and constraint tightening are updated online based on parameter learning and uncertainty set contraction. Second, the proposed framework avoids the commonly imposed assumption of the existence of a quadratically stabilizing linear feedback gain for the entire parametric uncertainty set by constructing terminal ingredients using the current parameter estimate. Further, we establish recursive feasibility, robust exponential stability of the origin, and boundedness of all closed-loop signals for the resulting adaptive tube MPC scheme.

\textit{Notations:} For two sets $\mathbb{A},\;\mathbb{B}\subseteq\mathbb{R}^{n}$, the Minkowski sum $\mathbb{A}\oplus\mathbb{B}\triangleq \{a+b\;|\;a\in\mathbb{A},\;b\in\mathbb{B}\}$ and $a\oplus\mathbb{B}\triangleq \{a+b\;|\; b\in\mathbb{B}\},$ where $a\in\mathbb{R}^{n}$, the Pontryagin difference $\mathbb{A}\ominus\mathbb{B}\triangleq\{a\;|\;a+b\in\mathbb{A}\;\forall b\in\mathbb{B}\}$, the set multiplication $\mathbb{A}\mathbb{B}\triangleq \{ ab\;|\;a\in\mathbb{A},\;b\in\mathbb{B}  \}$, where $a$, $b$ are of conformal dimensions. The convex hull of all elements in $\mathbb{B}$ is denoted by $\text{\textbf{co}} (\mathbb{B})$.  The notation $W\succ 0$ $(\succeq 0)$ implies the matrix $W$ is symmetric positive-definite (positive semi-definite). $||\cdot||_\infty$, $||\cdot||_2$ and $||\cdot||_F$ denote $\infty$-norm, $2$-norm and Frobenius norm, respectively, and $||b||^2_W\triangleq b^\top W b$ for a vector $b$. The notation $a_{i|t}$ represents the value of $a$ at time $t+i$ predicted at time $t$, and $(\cdot)^*$ on any term denotes its optimal value. A sequence $\{c(p),\;c(p+1),\;\dots,c(q-1),\;c(q)\}$ is represented as $\{c(i)\}_{i=p:q}$, and the set of all integers from $a$ to $b$ by $\mathbb{I}_a^b$. Any signal belonging to $\mathcal L_\infty$ implies it is bounded. A sequence vector $z_t$ is said to belong to $\mathcal{S}(m)$ if $\sum_{i=t}^{t+k}z^{\intercal}_iz_i\leq c_0mk+c_1$ $\forall \;t\in\mathbb{I}_1^\infty$, a given constant $m\geq 0$, and some $k\in\mathbb{I}_1^\infty$, where $c_0,\;c_1\geq0$ \cite[Theorem~4.11.2]{ioannou2006adaptive}.

\section{Problem Statement}\label{SecPS}
Consider the constrained discrete-time LTI system subject to external disturbance
\begin{align}
   & x_{t+1}=Ax_t+Bu_t+d_t, \label{osys1}\\
  &  x_t\in\mathbb{X}\text{ , }\;u_t\in\mathbb{U} \hspace{0.4cm}\forall t\in\mathbb{I}_0^\infty, \label{hc1}
\end{align}
where $x_t\in\mathbb{R}^n$, $u_t\in\mathbb{R}^m$, and $d_t\in\mathbb{R}^n$ denote the state, input, and external disturbance, respectively, at time $t$. The disturbance $d_t$ and the system parameter $\psi\triangleq \begin{bmatrix} A&B \end{bmatrix}\in\mathbb{R}^{n\times (n+m)}$ are unknown but belong to sets $\mathbb{D}$ and $\Psi$, respectively. The sets $\mathbb{X}$, $\mathbb{U}$ and $\mathbb{D}$ are known convex polytopes containing their origins in their respective interiors. The set $\Psi\subseteq \mathbb R^{n\times (n+m)}$ is the convex hull of known vertices $\psi^{[1]},\psi^{[2]},...,\psi^{[L]}$, where $L$ is a known finite positive integer. Further, for each $\hat\psi\triangleq\begin{bmatrix}
        \hat A & \hat B
    \end{bmatrix}\in\Psi$, $\exists$ a pair $(P,\;K)$ specific to $\hat \psi$ such that $P\succ 0$, and
    \begin{align}\label{dare1}
        P-(\hat A+\hat B K)^\top P (\hat A +\hat B K)-Q - K^\top R K \succeq 0,
    \end{align} for some given matrices $Q, R\succ 0$, where $P,\;Q\in\mathbb R^{n\times n}$, $R\in\mathbb R^{m\times m}$ and $K\in\mathbb R^{m\times n}$. This condition implies that each system corresponding to $\hat \psi\in\Psi$ is stabilizable and admits a stabilizing linear feedback gain $K$ with a quadratic Lyapunov function $x^\top P x$, which is a standard assumption in MPC design \cite{kouvaritakis2016model,rakovic2012homothetic,langson2004robust}.

The \emph{objective} is to design a suitable control $u_t$ that drives the state of the system \eqref{osys1} to its origin while ensuring that the hard constraints (\ref{hc1}) are satisfied despite the presence of the disturbance $d_t$. For the known dynamics ($A$ and $B$ are known) and disturbance-free case, classical approaches \cite{kouvaritakis2016model} could be used to solve the COCP $\mathbb P1$ in a receding horizon fashion.
\begin{align*}
    \mathbb{P}1:\;\;\;\;\;\;\min_{\mu_t}\;\sum_{i=0}^{N-1} \left(||x_{i|t}||^2_Q + ||u_{i|t}||^2_R \right)+||x_{N|t}||^2_P\;\;\; &\\
    \text{subject to}\;x_{0|t}=x_t, \;\;x_{i+1|t}=Ax_{i|t}+Bu_{i|t}\;\forall i\in\mathbb{I}_0^{N-1},&\\
    x_{i|t}\in\mathbb{X},\;\;u_{i|t}\in\mathbb{U}\;\;\forall i\in\mathbb{I}_0^{N-1},\;\;x_{N|t}\in\mathbb{X}_{TS}\subseteq \mathbb{X},&
\end{align*}
where $\mu_t\triangleq \{u_{i|t}\}_{i=0:N-1}$, and ${\mathbb{X}}_{{TS}}$ is the terminal set (refer to \cite[Ch.~2]{kouvaritakis2016model} for details). 

The classical COCP $\mathbb P1$ is not directly solvable in the presence of uncertainties in $A$, $B$ and $d_t$, and therefore requires reformulation. In the subsequent sections, motivated by set membership-based identification \cite{boydset,sasfi2023robust,lorenzen2019robust,lu2023robust} and homothetic tube-based MPC methods \cite{langson2004robust,rakovic2012homothetic}, we reformulate the COCP using an adaptive tube framework that is designed to reduce the uncertainty in $\psi$ at each time instant. 

% In the proposed method, we assume that the parametric uncertainty set $\Psi$ satisfies the following standard assumption which is useful for constraint tightening \cite{kohlerrampc,luflexitube,lorenzen2019robust,anch,darup2016computation,lu2023robust,tranosrusso}. 
% \begin{assum}\label{as1}
%     The set $\Psi\subseteq \mathbb R^{n\times (n+m)}$, that contains the true parameter $\psi$, is the convex hull of known vertices  $\psi^{[1]},\psi^{[2]},...,\psi^{[L]}$, where $L$ is some finite positive integer, with each element in $\Psi$ being controllable. 
%     \iffalse
%     Additionally, $\Psi$ is quadratically stabilizable, i.e., for given user-defined matrices $Q,\;R \succ 0$, $\exists$ $(P,\;K)$, where $P\succ 0\in\mathbb R^{n\times n}$ and $K\in\mathbb R^{m\times n}$ such that $\forall \begin{bmatrix}
%         \hat A & \hat B
%     \end{bmatrix}\in\Psi$
%     \begin{align}\label{quadPK}
%         P-(\hat A+\hat BK)^\top P (\hat A+ \hat B K)-Q- K^\top R K\succeq 0.
%     \end{align}
%     \fi
% \end{assum}
% \iffalse
% Assumption \ref{as1} is useful for constraint tightening, terminal set construction, and guaranteeing recursive feasibility and stability/convergence.
% \fi
Due to hard constraints on the system's states and inputs, the adaptive tube framework is equipped with a robust fallback mechanism to ensure constraint satisfaction regardless of parameter learning. The robust tube-based method serves as a backup if adaptation renders the modified COCP infeasible. We therefore start by developing the robust framework.  

\section{COCP components without adaptation}
A major motivation for this work is to relax the requirement of a common quadratically stabilizing linear feedback gain for the parametric uncertainty set $\Psi$, a common assumption in tube-based MPC frameworks for systems with uncertain parameters (see \cite{ lorenzen2019robust,lu2023robust,anchita2025,anch,buj2021,jafari2017adaptive,dhar2021indirect,tranosrusso,kohlerrampc,luflexitube}). To this end, we rewrite the system dynamics in \eqref{osys1} as follows.
\begin{align}\label{osys2}
    &x_{t+1}=\hat{A}_tx_t+\hat{B}_t u_t+\underbrace{{(A-\hat A_t)}x_t+{(B-\hat B_t)}u_t+d_t}_{=:w_t},
\end{align}
where the parameters $\hat A_t$, $\hat B_t$ are the estimates of $A$, $B$ obtained at time $t$ using a subsequently designed adaptive law and $\begin{bmatrix}
    \hat A_t & \hat B_t
\end{bmatrix}$ belong to the updated set $\Psi_t$. The term $w_t$ acts as a lumped additive disturbance. In the absence of parameter adaptation, $\hat A_t= \hat A_0$, $\hat B_t= \hat B_0$ and $ \Psi_t= \Psi_0$ $\forall t\in\mathbb I_0^\infty$, where $\Psi_0\triangleq \Psi$. Throughout this section, we use the suffix $t$ for easy extension to the adaptive case. %later express psi set with vertices

\subsection{Disturbance set characterization}
The knowledge of $\mathbb X$, $\mathbb U$, $\mathbb D$, $\Psi_t $ along with $\hat A_t$ and $\hat B_t$ is used to characterize the set $\mathbb W_t$ containing every possible value of $w_t$ $\forall t\in\mathbb I_0^\infty$.  Define 
\begin{subequations}
\begin{align}
   & \Phi_t \triangleq \left\{ \xi \in\mathbb{R}^{n\times{(n+m)}}\;\middle|\;\xi+\begin{bmatrix}
       \hat A_t & \hat B_t
   \end{bmatrix}\in \Psi_t \right\},\label{phit} \\
   &\Phi_{A_t} \triangleq \left\{ \xi\in\mathbb{R}^{n \times n}\;\middle|\;\begin{bmatrix}
       \xi & \eta
   \end{bmatrix}\in \Phi_t,\;\eta \in\mathbb{R}^{n\times m}  \right\},\label{phiAt}\\
 & \Phi_{B_t}  \triangleq \left\{ \xi\in\mathbb{R}^{n \times m}\;\middle|\;\begin{bmatrix}
       \eta & \xi
   \end{bmatrix}\in\Phi_t,\;\eta\in\mathbb{R}^{n\times n}  \right\}.\label{phiBt}
\end{align} \end{subequations}
Using definitions \eqref{phiAt} and \eqref{phiBt}, the set $\mathbb W_t$ is given by 
\begin{align}\label{Wmagset}
    \mathbb W_t\triangleq \Phi_{A_t}\mathbb X \oplus \Phi_{B_t} \mathbb U \oplus \mathbb D.
\end{align}
% which is also a convex polytope containing its origin in the interior.
A less conservative characterization of the sets containing $w_t$ is obtained by leveraging forward reachable sets of the state $x_t$ in the subsequent $N$ steps\footnote{We focus on only $N$ steps since the sets for $w_t$ are used in constraint tightening for the homothetic tube in a COCP with a prediction horizon $N$.}. At time $t$, the state measurement $x_t$ is used to compute the sets for $x_{t+1}, \; x_{t+2},\;x_{t+3},...,\;x_{t+N}$ by propagating the dynamics \eqref{osys1}. Since the sets for $x_{t+i}$ $\forall i\in\mathbb{I}_1^N$ depend on the current state $x_t$, we denote these as $\mathbb X_{i,t}$ where $i$ is the number of forward steps ahead of the current time $t$. The sets are constructed recursively as follows.
\begin{subequations}
\begin{align}
&\mathbb{X}_{0,t}\triangleq \{x_t\}, \label{Xreachset0}\\
   & \mathbb{X}_{i+1,t}\triangleq \left(\Psi_{A_t} \mathbb{X}_{i,t} \oplus \Psi_{B_t} \mathbb{U}\oplus\mathbb D \right) \cap \mathbb{X},\label{Xreachset_all}
   \end{align}
   \end{subequations}
where
\begin{subequations}
\begin{align}
&\Psi_{A_t}\triangleq \left\{ \xi\in\mathbb{R}^{n \times n}\;\middle|\;\begin{bmatrix}
       \xi & \eta
   \end{bmatrix}\in \Psi_t, \; \eta\in\mathbb{R}^{n\times m}  \right\},  \label{psiat}\\
   &\Psi_{B_t}\triangleq \left\{ \xi\in\mathbb{R}^{n \times m}\;\middle|\;\begin{bmatrix}
       \eta & \xi
   \end{bmatrix}\in \Psi_t, \; \eta\in\mathbb{R}^{n\times n}  \right\}.\label{psibt}
\end{align}
\end{subequations}
 We only require $\mathbb X_{i,t}$ $\forall i\in \mathbb{I}_0^{N-1}$. Using \eqref{phit}-\eqref{phiBt}, \eqref{Xreachset0}-\eqref{psibt}, the sets containing $w_{t+i}$ are given by
\begin{align}\label{WSet_all}
    \mathbb{W}_{i,t}\triangleq \Phi_{A_t} \mathbb X_{i,t}\oplus  \Phi_{B_t} \mathbb{U} \oplus \mathbb{D}\;\;\;\forall i\in\mathbb{I}_0^{N-1},
\end{align}
which are used to make robust state predictions in the MPC COCP.

\subsection{Terminal set construction}
Without adaptation, $\mathbb W_t\equiv \mathbb W_0$ and $\mathbb W_{i,t}\equiv \mathbb W_{i,0}$. While $\mathbb W_{i,t}$ contains the possible values of $w_{t+i}$ $\forall i\in\mathbb I_0^{N-1}$ starting at any given $t$, the set $\mathbb W_t$ is a superset of $\mathbb W_{i,t}$ and contains the possible values of $w_{t+i}$ $\forall i\in\mathbb I_0^\infty$. We, therefore, use $\mathbb W_t=\mathbb W_0$ to compute a suitable terminal set $\mathbb X_{TS_t}$ for formulating the COCP \cite{rakovic2012homothetic}, provided the following standard assumption holds.
\begin{assum}\label{Asdare1}
    Given constraint sets $\mathbb X$, $\mathbb U$, and $\mathbb W_0$, 
for each $\begin{bmatrix}
    \hat A & \hat B
\end{bmatrix}\in \Psi$, with the corresponding feedback gain 
$K(\hat A,\hat B)$ obtained from \eqref{dare1}, $\exists$ a non-empty 
terminal set $\mathbb X_{TS}(\hat A,\hat B)$ such that
\begin{subequations}\label{marpi1}
\begin{align}
&\mathbb X_{TS} \subseteq \mathbb X, \quad 
K\mathbb X_{TS}\subseteq \mathbb U, \\
&(\hat A+\hat B K)\,\mathbb X_{TS}
\oplus \mathbb W_0 \subseteq \mathbb X_{TS},
\end{align}
\end{subequations}
and $\mathbb X_{TS}$ contains the origin in its interior.
\end{assum}
The ideal choice of $\mathbb X_{TS_t}=\mathbb X_{TS_0}$ is the maximal admissible RPI set satisfying \eqref{marpi1}; this can be computed following \cite{dey2024computation} for $\hat A_t=\hat A_0$, $\hat B_t=\hat B_0$ and $K_t=K_0$. Note that the terminal set $\mathbb X_{TS_t}$ along with $K_t$ obtained from \eqref{dare1} are specific to a given parameter estimate ($\hat A_0$, $\hat B_0$), thereby relaxing the requirement for a common pair $(P_{com},\;K_{com})$ and a common terminal set for all the elements in $\Psi$, unlike \cite{lorenzen2019robust,lu2023robust,anchita2025,anch,buj2021,jafari2017adaptive,dhar2021indirect,tranosrusso,kohlerrampc,luflexitube}. Even when a common stabilizing feedback gain exists, constructing a common terminal set satisfying \eqref{marpi1} for all admissible systems can be difficult or may lead to excessive conservatism. This further motivates the use of terminal ingredients tailored to the current parameter estimate.

\subsection{Homothetic tube parameterization}
The tubes for state and control input obtained by solving the COCP at each time $t$ are denoted by 
    \begin{align}\label{CalTube}
   \mathcal{T}^{x}_t\triangleq \{ \mathbb{T}^{x}_{i|t} \}_{i=0:N}\;\text{ and }\;  \mathcal{T}^{u}_t\triangleq \{ \mathbb{T}^{u}_{i|t}\}_{i=0:N-1},
\end{align}
respectively. For the ease of solving the COCP, the tube-sections in $\mathcal{T}_t^{x}$ are defined as
\begin{align}
    \mathbb{T}^{x}_{i|t}\triangleq  \alpha_{i|t}\oplus \beta_{i|t}\mathbb{S}_t\subseteq {\mathbb{X}_{i,t}}\;\;\;\forall i\in\mathbb{I}_0^N ,\label{Tube1}
\end{align}
where $\alpha_{i|t}\in\mathbb{R}^n$ is the center, $\beta_{i|t}\geq 0$ is the scaling factor, and $\mathbb{S}_t$ determines the tube geometry. The set $\mathbb S_t$ is defined to be the convex hull of user-defined vertices $s^{[1]}_t,\;s^{[2]}_t,...,s^{[M_t]}_t$, where $M_t$ is finite. The state tube-sections can also be expressed as
\begin{subequations}\label{abarTube}
\begin{align}
     &  \mathbb{T}^{{x}}_{i|t}=\text{\textbf{co}} \left( \left\{ z^{[j]}_{i|t}\right\}_{j=1:M_t} \right), \text{ where} \label{Tube2}\\
&  z^{[j]}_{i|t}\triangleq \alpha_{i|t}+\beta_{i|t}s^{[j]}\;\;\;\forall (i,j,t)\in\mathbb{I}_0^N\times\mathbb{I}_1^{M_t}\times\mathbb{I}_0^\infty. \label{salphabeta}
\end{align}
\end{subequations}
Further, to achieve desirable guarantees of stability and recursive feasibility, a suitable choice of $\mathbb S_t$ \cite{rakovic2012homothetic,anchita2025} is that it contains the origin in its interior and is the outer RPI approximation of the minimal RPI set satisfying
\begin{align}\label{Stcompute}
    (\hat A_t+\hat B_t K_t)\mathbb S_t \oplus \mathbb W_t\subseteq \mathbb S_t.
\end{align}
Here, $K_t$ is the feedback gain obtained by solving \eqref{dare1} for $\hat A_t$, $\hat B_t$, and $\mathbb S_t$ can be computed following \cite{rakovic2005invariant}. Again, without any adaptation, we have $\mathbb S_t= \mathbb S_0$ $\forall t\in\mathbb I_0^\infty$.

Corresponding to \eqref{Tube2}, the control tube-sections are written as
\begin{align}
    \mathbb{T}^{u}_{i|t}\triangleq   \left\{ v^{[j]}_{i|t}\right\}_{j=1:M_t}  \;\;\forall (i,t)\in\mathbb{I}_0^{N-1}\times\mathbb{I}_0^\infty. \label{Tube3}
\end{align}
Since the tube-sections are expressed as convex hulls of vertices, any point $\eta\in\mathbb{T}^{{x}}_{i|t}$, can be expressed as $ \eta=\sum_{j=1}^{M_t} \tau^{[j]}_{(i,\;t)}z_{i|t}^{[j]}, \text{ where }\tau^{[j]}_{(i,\;t)}\in [0,1]$ and $ \sum_{j=1}^{M_t} \tau_{(i,\;t)}^{[j]}=1$.
% \begin{align}\label{convX}
%     \eta=\sum_{j=1}^{M_t} \tau^{[j]}_{(i,\;t)}z_{i|t}^{[j]}, \text{ where }\tau^{[j]}_{(i,\;t)}\in [0,1], \;\sum_{j=1}^{M_t} \tau_{(i,\;t)}^{[j]}=1
% \end{align}
The corresponding control input $u$ as a function of $\eta$ is given by the following function
\begin{align}\label{convU}
    u=\mathbf{u}(\eta,\mathbb T^x_{i|t},\mathbb T^u_{i|t})\triangleq\sum_{j=1}^{M_t}\tau^{[j]}_{(i,\;t)}v_{i|t}^{[j]}.
\end{align}
% This concludes the design of all the ingredients required to formulate and run a homothetic tube-based MPC COCP in the absence of any adaptation. 

\section{Adaptive tube MPC framework}
In this section, we introduce parameter learning and reformulate the COCP with recursive feasibility and stability guarantees. The dynamics used for developing the framework is given by \eqref{osys2} 
\begin{align*}
    x_{t+1}=\hat A_t x_t +\hat B_t u_t +w_t,%\tag{\ref{osys2}}
\end{align*}
where $w_t\in\mathbb W_t$ defined in \eqref{Wmagset}. The suffix $t$ was already introduced with the estimated parameters and sets characterized in the previous section to indicate their time-varying nature due to adaptation. Below, we provide the laws for carrying out the adaptation of the parameter estimates $\hat \psi_t\triangleq\begin{bmatrix} \hat A_t & \hat B_t\end{bmatrix}$ and the uncertainty set $\Psi_t$.

\subsection{Parameter learning and uncertainty set refinement}
At each time step $t\in\mathbb I_1^\infty$, the available data $x_{t}$, $x_{t-1}$, $u_{t-1}$ are leveraged to compute a non-falsified set for the true parameter $\psi=\begin{bmatrix}
    A&B
\end{bmatrix}$ as follows \cite{boydset} 
\begin{align}\label{nonfset}
    \Xi_{t}\triangleq \left\{ \begin{bmatrix}\hat A & \hat B\end{bmatrix}\;\middle|\; x_{t}-\hat Ax_{t-1}-\hat Bu_{t-1}\in\mathbb{D} \right\}.
\end{align}
The parametric uncertainty set is then updated as
\begin{align}\label{updatepsi}
\Psi_{t}\triangleq \Psi_{t-1} \cap \Xi_{t}\;\;\;\forall t\in\mathbb I_1^\infty;\;\;\Psi_0 =\Psi.
\end{align}
\begin{lem}\label{lemma1}
   For the plant dynamics \eqref{osys1}, the true parameter $\psi=\begin{bmatrix}
    A & B
\end{bmatrix}\in\Psi_t$ $\forall t\in\mathbb I_1^\infty$, where the uncertainty sets $\Psi_t$, obtained using \eqref{nonfset} and \eqref{updatepsi}, are nested, i.e., $\Psi_t\subseteq\Psi_{t-1}$ $\forall t\in\mathbb I_1^\infty$. This further implies that $\Psi_t\neq\emptyset$ at any time $t$.
\end{lem}
{\textit{Proof:}} The proof easily follows from the plant dynamics \eqref{osys1} and the non-falsified set definition in \eqref{nonfset} and the recursive update law \eqref{updatepsi}, since both $\Psi_0$ and $\Xi_t$ $\forall t\in\mathbb I_0^\infty$ contain the true parameter $\psi$.\hfill$\blacksquare$

Since the updated uncertainty set may yield a new set of vertices, we add a time index to the vertices of $\Psi_t$. 
\begin{gather}
    \begin{aligned}
    &\Psi_t=\textbf{co}\left( \left\{ \psi_t^{[i]} \right\}_{i=1:L_t}  \right)\;\;\forall t\in\mathbb{I}_0^\infty, \;\text{with }\\
    & \psi^{[i]}_0\triangleq \psi^{[i]}\;\;\;\forall i\in\mathbb{I}_1^{L_0},\;\text{where }L_0\triangleq L.
\end{aligned} 
\end{gather}
Using the updated uncertainty set $\Psi_{t}$, obtained with \eqref{updatepsi}, we compute the point estimate $\hat \psi_{t}$ using a normalized gradient descent-based law along with a projection operator that projects the estimates onto the uncertainty set $\Psi_{t}$. To this end, we rewrite \eqref{osys1} as
\begin{align}\label{osys1_gd}
    x_t=\psi g_{t-1} +d_{t-1} \;\;\forall t\in\mathbb I_1^\infty,
\end{align}
where $g_{t-1}\triangleq \begin{bmatrix}
        x_{t-1}^\top &u_{t-1}^\top
    \end{bmatrix}^\top\in\mathbb R^{n+m}$. Here \eqref{osys1_gd} is in a linear regression form with a perturbation-like term $d_{t-1}$, and at any time $t$, we know the terms $x_{t}$ and $g_{t-1}$. Accordingly the update law is given by  \begin{gather}  \label{psiadapt}
\begin{aligned}
    &\bar \psi_{t}=\hat \psi_{t-1} + \kappa \frac{(x_t-\hat \psi_{t-1}g_{t-1})g_{t-1}^\top}{1+g^\top_{t-1}g_{t-1}},  \\
    & \hat \psi_{t}=\begin{cases}
         \bar \psi_{t},\;\;\text{if }\bar \psi_{t}\in\Psi_{t}\\
         \underset{\xi\in\Psi_{t}}{\arg \min} \;|| \xi-\bar \psi_{t}||_F,\;\;\text{otherwise},
    \end{cases}
\end{aligned}
\end{gather}~where $\kappa\in(0,2)$.

\begin{lem}\label{lemma2}
    Let $e_t\triangleq ({x_t-\hat \psi_{t-1}g_{t-1}})/({1+g^\top_{t-1}g_{t-1}})$. The update law in \eqref{psiadapt} with $\kappa\in(0,2)$ guarantees
    \begin{itemize}
        \item $e_t$, $ e_t(1+g^\top_{t-1}g_{t-1})^{\frac{1}{2}}$, $\hat\psi_t$ $\in\mathcal L_\infty$,
        \item $e_t$, $ e_t(1+g^\top_{t-1}g_{t-1})^{\frac{1}{2}}$, $||\hat\psi_t-\hat\psi_{t-1}||_F$ $\in \mathcal S(\bar d^2)$, where $\bar d$ is the upper bound of $||d_t (1+g^\top_{t-1}g_{t-1})^{-\frac{1}{2}}||_2$.
    \end{itemize}
\end{lem}
{\textit{Proof: }}The proof follows from \cite[Lemma~1]{anchita2025}, \cite[Theorem~4.11.4]{ioannou2006adaptive} with the use of Lemma \ref{lemma1}. \hfill $\blacksquare$

\subsection{Reformulated COCP with adaptation}
Due to the adaptation of $\Psi_t$ and the parameter estimate $\hat\psi_t = \begin{bmatrix} \hat A_t & \hat B_t \end{bmatrix}$, the system dynamics in \eqref{osys2} and the associated sets defined in \eqref{phit}-\eqref{Wmagset} and \eqref{Xreachset_all}-\eqref{WSet_all} are updated whenever $\Psi_t$ and/or $\hat\psi_t$ change. Accordingly, the matrices $P_t$, $K_t$, the terminal set $\mathbb X_{TS_t}$ and the polytope $\mathbb S_t$ defining the tube cross-sectional shape are recomputed to remain consistent with the updated uncertainty description. While this introduces additional computational effort, it enables the controller to exploit any reduction in the size of $\Psi_t$, leading to the following dual benefits.
\begin{itemize}
\item The terminal set, which is desirably the maximal admissible RPI set, increases in size, implying the need for reduced control effort to reach the terminal set, and 
\item The tube cross-sectional shape $\mathbb S_t$ reduces in size, allowing for a better representation of the uncertainty in the propagation of the true state. 
\end{itemize}
The use of updated parameter estimates and uncertainty sets yields a modified COCP at each time instant. However, the switching from the old setup at $t-1$ to the new COCP at $t$ with updated $(P_{t},\;K_{t})$ and terminal set $\mathbb X_{TS_t}$ may affect recursive feasibility and stability. The treatment of feasibility under such switching is discussed in the next subsection. To ensure stability guarantees, we impose certain conditions on the choice of $(P_t,K_t)$. 

\begin{crit}\label{Crdare2}
    Given the parameters $\hat A_t$, $\hat B_t$, the pair $(P_{t-1}, K_{t-1})$ and matrices $Q,\;R$, where $Q,\;R,\;P_{t-1}\succ 0$, and constraint sets $\mathbb X$, $\mathbb U$ and $\mathbb W_t$, choose a pair $(P_t,K_t)$ that satisfies $P_t\succ 0$ and the following conditions.
    \begin{subequations}
        \label{subeqnTS}
        \begin{align}
    \text{(a)}\;\;&    P_t-(\hat A_t+\hat B_t K_t)^\top P_t (\hat A_t+\hat B_t K_t)\nonumber\\
    &\;\;\;\;\;\;\;\;\;\;\;\;\;\;\;\;\;\;\;\;\;\;\;\;\;\;\;\;-Q-K_t^\top R K_t \succeq 0 \label{darecri1}\\
\text{(b)}\;\;&    P_{t-1}-(\hat A_t+\hat B_t K_t)^\top P_t (\hat A_t+\hat B_t K_t)\nonumber\\
    &\;\;\;\;\;\;\;\;\;\;\;\;\;\;\;\;\;\;\;\;\;\;\;\;\;\;\;\;-Q-K_{t-1}^\top R K_{t-1} \succeq 0, \label{dare2cri}
    \end{align}
    \end{subequations}
    and, (c) $\exists$ a non-empty terminal set $\mathbb X_{TS_t}$ such that
    \begin{subequations}
        \begin{align}
        &\mathbb X_{TS_t}\subseteq\mathbb X,\;\;K_t\mathbb X_{TS_t}\subseteq\mathbb U,\tag{\ref{subeqnTS}c}\label{2ndterminala}\\
&        (\hat A_t+\hat B_t K_t)\mathbb X_{TS_t}\oplus \mathbb W_t\subseteq \mathbb X_{TS_t} \tag{\ref{subeqnTS}d},\label{2ndterminalb}
    \end{align}
    \end{subequations} and $\mathbb X_{TS_t}$ contains the origin in its interior.
\end{crit}

Condition (b) plays a crucial role in preserving stability under adaptation. Specifically, it ensures that the Lyapunov function, which is associated with the cost matrices $Q$, $R$, and $P_{t-1}$ at time $t-1$ and $P_t$ at time $t$, does not increase when the prediction model switches from $(\hat A_{t-1}, \hat B_{t-1})$ to $(\hat A_t, \hat B_t)$ and the feedback gain is updated from $K_{t-1}$ to $K_t$. In other words, even though the terminal cost and prediction dynamics are modified, the value function remains non-increasing along system trajectories, thereby preserving the stability guarantee across switching instances. In the absence of adaptation, i.e., when $(\hat A_t, \hat B_t) = (\hat A_{t-1}, \hat B_{t-1})$, Condition~(b) reduces to Condition~(a), and the requirement coincides with the standard discrete-time Lyapunov inequality for classical MPC \cite[Ch.~2]{kouvaritakis2016model}.

Conditions (a) and (c) follow from \eqref{dare1} as a property of $\Psi_0$ and Assumption~\ref{Asdare1}, respectively. Unlike the quadratic stabilizability assumption, which requires the existence of a common pair $(P_{{com}}, K_{{com}})$ satisfying $P_{{com}} - (\hat A + \hat B K_{{com}})^\top P_{{com}} (\hat A + \hat B K_{{com}})- Q - K_{{com}}^\top R K_{{com}} \succeq 0$, $\forall \begin{bmatrix}    \hat A & \hat B \end{bmatrix}\in \Psi$, together with a common terminal set, Condition~(b) only enforces a one-step compatibility between consecutive parameter updates. Thus, it avoids the need for a common Lyapunov function and feedback gain that are valid for the entire uncertainty set.

Provided Criterion \ref{Crdare2} holds, we compute the polytope $\mathbb S_t$ as the outer RPI approximation of the minimal RPI set satisfying \eqref{Stcompute}. 

\begin{rem}
If Criterion~\ref{Crdare2} does not hold, the point estimates $\hat A_t,\hat B_t$ and the pair $(P_t,K_t)$ remain unchanged. Nevertheless, adaptation may still occur through contraction of the uncertainty set.
\end{rem}

% To this end, we make the next assumption.
%  \begin{assum}\label{asKSet}
%    For given cost function weights $Q$ and $R$, the LQR gains obtained corresponding to each of the elements in the set $\Psi$ can be bounded by a known convex set $\mathbb{K}$ with a finite number of vertices $M$, i.e.,
%    \begin{align}
%   & \forall  \begin{bmatrix}
%        \hat A & \hat B
%    \end{bmatrix}\in\Psi, \;  - \left( R+\hat B^\top \hat P \hat B \right)^{-1}     \hat B^\top \hat P \hat A \in\mathbb K\triangleq  \nonumber\\
%    &\;\;\;\;\;\;\textbf{co}( \{ K^{[1]},K^{[2]},\;...,\;K^{[M]} \})
%    \end{align}
%    where $\hat P$ is the solution of the discrete algebraic Riccati equation for $\hat A$, $\hat B$. 
% \end{assum}
% \subsection{Reformulated COCP}

The COCP for adaptive tube MPC with decision variable $\bar \mu_t\triangleq \left\{ \{(\alpha_{i|t}, \beta_{i|t})\}_{i=0:N}, \left\{v_{i|t}^{[j]}\right\}_{i=0:N-1, j=1:M_t} \right\}$ is given by 
\begin{subequations}\label{MPC2}
\begin{align}
   & \mathbb{P}2_t:\;\min_{\bar\mu_t}\;\sum_{j=1}^{M_t} \left(\sum_{i=0}^{N-1} \left(|| z_{i|t}^{[j]}||^2_Q + || v_{i|t}^{[j]}||^2_R \right)+|| z_{N|t}^{[j]}||^2_{P_t} \right)\label{MPC2cost}\\
    &\text{subject to}\;\;\eqref{CalTube}-\eqref{abarTube}\nonumber\\
  &\beta_{i|t}\geq 0\;\;\;\forall i\in\mathbb{I}_{1}^{N},  \label{cons1m}\\
  & \mathbb T^x_{0|t}=\{x_t\}, \label{cons6m}\\
   % & \alpha_{N|t}=  ( \hat{A}_t+\hat{B}_tK_t ) \alpha_{N|t-1},\;\; \;t\in\mathbb{I}_1^\infty,  \tag{\ref{MPC2}b}\label{cons7m} \\
   % & \beta_{N|t}\leq \rho_{1_t}\beta_{N|t-1}+\rho_{2_t}, \;\;\;t\in\mathbb{I}_1^\infty , 
   % & \beta_{N|t}\mathbb S_t\subseteq (\hat A_{t}+\hat B_{t}K_t)\beta_{N|t-1}\mathbb S_t\oplus\mathbb W_t, \;\;\;t\in\mathbb{I}_1^\infty , \tag{\ref{MPC2}c}\label{cons8m} \\ 
  &\mathbb{T}^{{x}}_{i|t}\subseteq {\mathbb{X}},\text{ }\mathbb{T}^u_{i|t}\subseteq\mathbb{U}\;\;\;\forall i\in\mathbb{I}_{0}^{N-1} , \label{cons2m}\\ 
   & \mathbb{T}^{{x}}_{N|t}\subseteq {\mathbb{X}_{TS_t}}\subseteq \mathbb X, \text{ and}  \label{cons3m}\\
   % \end{align}\newpage\noindent
   % \begin{align}
   % & \mathbb T^x_{N|t}\subseteq (\hat A_{t}+\hat B_{t}K_t)\mathbb T^x_{N|t-1}\oplus\mathbb W_t, \;\;\;t\in\mathbb{I}_1^\infty , \label{cons4m}\\
    &\hat{A}_t z_{i|t}^{[j]}+\hat{B}_t v_{i|t}^{[j]}\in \mathbb{T}^{{x}}_{i+1|t}\ominus \mathbb{W}_{i,t} \; \forall (i,j)\in\mathbb{I}_{0}^{N-1}\times\mathbb{I}_1^{M_t}.  \label{cons5m}
\end{align}
\end{subequations}
Using the solution of the COCP and \eqref{convU}, we compute the input $u_t^*= \mathbf{u}(x_t,\mathbb T^{x^*}_{i|t},\mathbb T^{u^*}_{i|t})$ that is applied to the plant \eqref{osys1} to evolve to the next state $x_{t+1}$. The COCP is labelled as $\mathbb P2_t$ since the formulation depends on time $t$. 

It is possible that the new set $\Psi_t$, with the new estimates $\hat A_t, \;\hat B_t$ and the consequently computed sets, leads to an infeasible COCP. 
% It is important to note that there are possibilities for the new set $\Psi_t$ with the new estimates $\hat A_t, \;\hat B_t$ and the consequently computed sets to lead to an infeasible COCP. 
First, there may not exist a $(P_t,K_t)$ that satisfies Criterion \ref{Crdare2}. Second, even if Criterion \ref{Crdare2} holds, \eqref{MPC2} may be infeasible due to the inadequate length of the prediction horizon. In either case, we revert to the previous estimates by setting $\hat A_t\leftarrow \hat A_{t-1}, \hat B_t\leftarrow \hat B_{t-1}$, $P_t\leftarrow P_{t-1}$, $K_t\leftarrow K_{t-1}$ and 
\begin{align}\label{backpsi}
    \Psi_t\leftarrow \textbf{co}\left(\Psi_t, \left\{ \begin{bmatrix}
        \hat A_{t-1} & \hat B_{t-1}
    \end{bmatrix}\right\}\right),
\end{align}
and compute the sets in \eqref{phit}-\eqref{Wmagset}, \eqref{Xreachset_all}-\eqref{WSet_all}, $\mathbb X_{TS_t}$ and $\mathbb S_t$. At any time step, if there is a change in the set $\Psi_t$ or the estimates $\hat A_t$ $\hat B_t$, the shape of the disturbance-related sets along with the polytope $\mathbb S_t$ changes. This renders the tube an \textit{adaptive} nature, allowing for improved characterization of the uncertainties in trajectory propagation of the true system \eqref{osys1}. Algorithm \ref{algoatube} outlines the steps for implementing the proposed framework.

\begin{algorithm}[ht!]
\caption{Adaptive Tube MPC}\label{algoatube}
\begin{algorithmic}[1]
 \REQUIRE $\mathbb{X}$, $\mathbb{U}$, $\mathbb{D}$, $\psi^{[i]}$ $\forall i\in\mathbb{I}_1^{L_0}$, $N$, $Q$, $R$, $\hat \psi_0$, $\kappa$.
 \ENSURE  $\bar\mu_t^*$ $\;\forall t\in\mathbb{I}_0^\infty$.\\
\textbf{Steps:}\\
\STATE Initialize $\Psi_0=\Psi$, i.e., $\psi^{[i]}_0=\psi^{[i]}$ $\forall i\in\mathbb{I}_1^{L_0}$, a constant $c_{backup}=0$ and time $t=0$. 
\STATE Measure $x_0$.
\WHILE{$t\geq 0$}
\STATE Compute $\Phi_{t}, \Phi_{A_{t}}, \Phi_{B_{t}}, \mathbb W_{t}$, $\Psi_{A_t}$, $\Psi_{B_t}$, and $\mathbb X_{i,t}$, $\mathbb W_{i,t}$ $\forall i\in\mathbb I_0^{N-1}$ and $\mathbb S_{t}$ using \eqref{phit}-\eqref{WSet_all} and \eqref{Stcompute}.
\IF{$t==0$}
\STATE Choose $P_t$, $K_t$ and $\mathbb X_{TS_t}$ that satisfies \eqref{dare1} and Assumption \ref{Asdare1}. 
\ELSE
\STATE Check if $\exists$ $P_t$, $K_t$, $\mathbb X_{TS_t}$ satisfying Criterion \ref{Crdare2}.
\IF{Criterion \ref{Crdare2} does not hold}
\STATE Set $c_{backup}\leftarrow 1$, $t\leftarrow t-1$ and go to Step \ref{stepjump}.
\ENDIF
\ENDIF
% \STATE Compute $\mathbb S_{t+1}$ using \eqref{Stcompute}.\label{stepjump2}
\STATE Run the COCP $\mathbb P 2_t$.
\IF{$\bar\mu_t^*\neq\emptyset$, i.e., $\mathbb P 2_t$ is feasible}
 \STATE Set $c_{backup}\leftarrow0$.
 \ELSE
  \IF{$t==0$}
    \STATE Exit the algorithm (initially infeasible).
    \ELSE
    \STATE Set $c_{backup}\leftarrow1$ and $t\leftarrow t-1$.
    \ENDIF
\ENDIF
\IF{$c_{backup}==0$} 
  \STATE Apply $u_t^*= \mathbf{u}(x_t,\mathbb T^{x^*}_{i|t},\mathbb T^{u^*}_{i|t})$ to the plant \eqref{osys1}.
  \STATE Measure $x_{t+1}$ from \eqref{osys1}.
  \STATE Compute the non-falsified set $\Xi_{t+1}$, the uncertainty set $\Psi_{t+1}$, and the point estimate $\hat \psi_{t+1}$ using \eqref{nonfset}, \eqref{updatepsi}, and \eqref{psiadapt}, respectively.
\ENDIF  
%   \STATE Check if $\exists$ $P_{t+1}$, $K_{t+1}$ and $\mathbb X_{TS_{t+1}}$ satisfying Criterion \ref{Crdare2}.
%   \IF{Criterion \ref{Crdare2} does not hold}
%      \STATE Set $c_{backup}\leftarrow1$, and go to Step \ref{stepjump}.
%   \ENDIF
% \ENDIF
\IF{$c_{backup}==1$} \label{stepjump}
\STATE Compute $\Psi_{t+1}$ using \eqref{backpsi}, and set $\hat A_{t+1}\leftarrow\hat A_t$, $\hat B_{t+1}\leftarrow\hat B_t$, $\hat P_{t+1}\leftarrow\hat P_t$, $\hat K_{t+1}\leftarrow\hat K_t$. \label{stepbackup}
\ENDIF\label{stepjump2}
% \STATE Compute $\mathbb S_{t+1}$ using \eqref{Stcompute}.\label{stepjump2}
\STATE Update $t\leftarrow t+1$.
\ENDWHILE
\end{algorithmic}
\end{algorithm}
\begin{rem}
    It is possible that the number of vertices of $\Psi_t$ may increase with $t$. To ensure that the approach is tractable, one can either stop the set adaptation after the number of vertices $L_t$ reaches a pre-decided upper bound depending on the computational resources or opt for methods discussed in \cite{lorenzen2019robust,CHISCI20011019,Veres01011999}. 
\end{rem}
\begin{rem}
According to Algorithm~\ref{algoatube}, if $\mathbb P2_t$ is infeasible at some time $t\in\mathbb I_1^\infty$, the COCP is reformulated using the setup described in Steps~\ref{stepjump}-\ref{stepjump2} and solved with the updated formulation. To avoid the additional computational effort associated with resolving the modified problem at time $t$, the control input $u_t$ can instead be computed using the previously obtained optimal solution at time $t-1$, together with the measured state $x_t$ and \eqref{convU}.
\end{rem}
% \begin{rem}
% A complete theoretical characterization of the existence of a pair $(P_t,K_t)$ satisfying \eqref{darecri1} and \eqref{dare2cri}, as well as precise conditions under which such a pair is guaranteed to exist, is not studied in this work. Instead, we adopt a heuristic procedure for computing $(P_t,K_t)$. As long as the uncertainty set $\Psi_t$ does not undergo a significant reduction, the previously computed pair $(P_{t-1},K_{t-1})$ is retained. When $\Psi_t$ shrinks sufficiently due to informative data, we search for a new feedback gain and terminal cost matrix, where \eqref{darecri1} is solved temporarily with inflated weights $\bar Q \succ Q$ and $\bar R \succ R$. The inflation increases the margin in the Lyapunov inequality and facilitates satisfaction of \eqref{darecri1}. Once a feasible pair $(P_t,K_t)$ is obtained, it is verified that \eqref{dare2cri} also holds with the original weights $Q$ and $R$. The inflated weights are used solely as a computational aid and do not modify the final design parameters.
% \end{rem}

\subsection{Recursive Feasibility and Stability Analysis} 
Next, we prove that the proposed COCP $\mathbb P2_t$ following Algorithm \ref{algoatube} is recursively feasible, and leads to robust exponential stability of the origin of \eqref{osys1}.
\begin{thm}\label{theo1}
    Suppose Assumption \ref{Asdare1} holds, and the COCP $\;\mathbb{P}2_t$ is feasible at some time $t$. Then, Algorithm \ref{algoatube} ensures that $\;\mathbb P2_t$ is feasible at time $t+k$ $\forall k\in\mathbb I_1^\infty$.
\end{thm}
{\textit{Proof:}} The proof proceeds by induction. Assume that the COCP is feasible at time $t$. We show feasibility at time $t+1$.

Given a feasible solution at $t$, the plant \eqref{osys1} evolves and the parameter estimates, uncertainty set and disturbance sets are updated. Criterion~\ref{Crdare2} is then checked at $t+1$. If it is satisfied, new matrices $P_t$, $K_t$, terminal set and tube geometry $\mathbb S_t$ are computed, resulting in a modified COCP. Since the problem data change, feasibility of this newly formed problem cannot be guaranteed a priori. If $\mathbb P2_{t+1}$ is feasible, the algorithm proceeds with its solution. Otherwise, or if Criterion~\ref{Crdare2} does not hold, the backup setup in Step~\ref{stepbackup} is used. For the backup, two situations arise.

\begin{itemize}

\item[(i)] \emph{No update of point estimate or uncertainty set.} In this case, the COCP coincides with the standard homothetic tube-based MPC formulation. Recursive feasibility follows from the existing literature (see \cite{langson2004robust,rakovic2012homothetic,lorenzen2019robust,anchita2025}). Let the feasible solution at $t+1$ for the non-adaptive scenario be denoted as follows using $\underline{(\cdot)}$, with $(\cdot)^*$ denoting the optimal solution of $\mathbb P_t$.

\begin{subequations}\label{nonada}
    \begin{align}
    & \underline{\mathbb T^x_{0|t+1}}=\{x_{t+1}\},\\
      &  \underline{\mathbb T^x_{i|t+1}}=\mathbb T^{x^*}_{i+1|t}\;\;\forall i\in\mathbb I_1^{N-1},\\
       & \underline{\mathbb T^u_{i|t+1}}=\mathbb T^{u^*}_{i+1|t}\;\;\forall i\in\mathbb I_0^{N-2},\\
      &  \underline{\mathbb T^x_{N|t+1}}=\underline{\alpha_{N|t+1}}\oplus \underline{\beta_{N|t+1}}\mathbb S_t,\\
      &   \underline{\mathbb T^u_{N-1|t+1}}=K_t \underline{\mathbb T^x_{N-1|t+1}}=K_t \mathbb T^{x^*}_{N|t},\\
      &\text{where }\underline{\alpha_{N|t+1}}=(\hat A_t+\hat B_t K_t)\alpha_{N|t}^* \text{ and }\\
      & \underline{\beta_{N|t+1}}=\min_{\beta} \left\{ \beta \;\middle|\; (\hat A_t+\hat B_t K_t)\beta_{N|t}^* \mathbb S_t\oplus \mathbb W_t \subseteq \beta\mathbb S_t   \right\},
             \end{align} \end{subequations}
with $ \underline{z_{i|t+1}^{[j]}}$ and $\underline{v_{i|t+1}^{[j]}}$ being the vertices of $\underline{\mathbb T^x_{i|t+1}}$ $\forall i\in\mathbb I_0^{N}$ and the elements of $\underline{\mathbb T^u_{i|t+1}}$ $\forall i\in\mathbb I_0^{N-1}$, respectively, $\forall j\in\mathbb I_1^{M_t}$. The solution is useful in proving recursive feasibility of the next case.

\item[(ii)] \emph{Uncertainty set update only.} 
Here, the point estimates $\hat A_t,\;\hat B_t$ and $(P_t,K_t)$ remain unchanged, while the uncertainty set is refined according to \eqref{backpsi}, yielding $\Psi_{t+1}\subseteq \underline{\Psi_{t+1}}=\Psi_t$. Consequently, $\Phi_{t+1}\subseteq \underline{\Phi_{t+1}}=\Phi_t,\;\Phi_{A_{t+1}}\subseteq\underline{\Phi_{A_{t+1}}}=\Phi_{A_t}, \;\Phi_{B_{t+1}}\subseteq \underline{\Phi_{B_{t+1}}}= \Phi_{B_t},\;\Psi_{A_{t+1}}\subseteq \underline{\Psi_{A_{t+1}}}=\Psi_{A_t},\; \Psi_{B_{t+1}}\subseteq \underline{\Psi_{B_{t+1}}}= \Psi_{B_t},$ resulting in
\begin{subequations}\label{WSTimprove}
\begin{align}
&\mathbb W_{t+1}\subseteq \underline{\mathbb W_{t+1}}=\mathbb W_t, \;\mathbb S_{t+1}\subseteq \underline{\mathbb S_{t+1}} =\mathbb S_t,\\
&\mathbb X_{i,t+1}\subseteq \underline{\mathbb X_{i,t+1}}\subseteq \mathbb X_{i+1,t}\;\;\forall  i\in\mathbb I_0^{N-1},\\
&\mathbb W_{i,t+1}\subseteq  \underline{\mathbb W_{i,t+1}}\subseteq\mathbb W_{i+1,t}\;\;\forall  i\in\mathbb I_0^{N-1},\\
& \mathbb X \supseteq \mathbb X_{TS_{t+1}}\supseteq\underline{\mathbb X_{TS_{t+1}}}=\mathbb X_{TS_t}.
\end{align}\end{subequations}
Leveraging \eqref{WSTimprove} and the solution \eqref{nonada} for the non-adaptive case, we prove that the following is a feasible solution for the COCP at $t+1$  under the refined sets.
\begin{subequations}\label{actualsolution}
    \begin{align}
&\mathbb T^x_{i|t+1}=\alpha_{i|t+1}\oplus\beta_{i|t+1}\mathbb S_{t+1} \;\;\;\forall i\in\mathbb I_0^N, \text{ where }\\
&\alpha_{0|t+1}=x_{t+1}, \beta_{0|t+1}=0 ,\\
&\alpha_{i|t+1}=\alpha_{i+1|t}^*,\;\beta_{i|t+1}=\beta_{i+1|t}^*\;\;\forall i\in\mathbb I_1^{N-1} ,\label{fs2}\\
&\alpha_{N|t+1}=\underline{\alpha_{N|t+1}}, \;\beta_{N|t+1}=\underline{\beta_{N|t+1}}  \label{fs4},\text{ and}\\
&\mathbb T^u_{i|t+1}= \left\{\mathbf{u}(z_{i|t+1}^{[j]},\underline{\mathbb T^{x}_{i|t+1}}, \underline{\mathbb T^u_{i|t+1}})\middle| j\in\mathbb I_1^{M_{t+1}}\right\}.
\end{align}
\end{subequations}
It is easily seen that the centers and scaling factors in \eqref{actualsolution} satisfy \eqref{CalTube}-\eqref{abarTube}, \eqref{cons1m} and \eqref{cons6m} at $t+1$. For the tube-sections, we can write
  \begin{align*}
 \mathbb T_{0|t+1}^x&=\underline{ \mathbb T_{0|t+1}^x}\subseteq  \mathbb X_{1,t},\end{align*}\begin{align*}
 \mathbb T_{i|t+1}^x&=\alpha_{i|t+1}\oplus\beta_{i|t+1}\mathbb S_{t+1} \nonumber\\
&\subseteq \alpha_{i+1|t}^*\oplus\beta_{i+1|t}^*\mathbb S_{t} \nonumber\\
&=\underline{\mathbb T^x_{i|t+1}}\subseteq \mathbb X_{i+1,t} \;\;\;\forall i\in\mathbb I_1^{N-1}, \\
\mathbb T^x_{N|t+1}&=\alpha_{N|t+1}\oplus\beta_{N|t+1}\mathbb S_{t+1} \nonumber\\
&\subseteq  \underline{\alpha_{N|t+1}}\oplus\underline{\beta_{N|t+1}}\mathbb S_{t} \nonumber\\
&=\underline{\mathbb T^x_{N|t+1}}\subseteq \mathbb X_{TS_t}\subseteq\mathbb X_{TS_{t+1}}\subseteq\mathbb X, \\
% \mathbb T^x_{N|t+1}&\subseteq \underline{\mathbb T^x_{N|t+1}}\subseteq(\hat A_t + \hat B_t K_t)\mathbb T^{x^*}_{N|t}\oplus \mathbb W_t \nonumber\\
% &= (\hat A_{t+1} + \hat B_{t+1} K_{t+1})\mathbb T^{x^*}_{N|t}\oplus \mathbb W_{t} \nonumber\\
% &\subseteq (\hat A_{t+1} + \hat B_{t+1} K_{t+1})\mathbb T^{x^*}_{N|t}\oplus \mathbb W_{t+1}.
\textbf{co}(\mathbb T^u_{i|t+1}) &\subseteq \textbf{co}(\underline{\mathbb T^u_{i|t+1}})\Rightarrow  \mathbb T^u_{i|t+1}\subseteq \mathbb U \;\;\forall i\in\mathbb I_0^{N-1},
\end{align*} 
 which together prove the satisfaction of \eqref{cons2m} and \eqref{cons3m} at time $t+1$.
Finally, for \eqref{actualsolution} to satisfy \eqref{cons5m}, note that from the non-adaptive case we have
\begin{align*}
 \hat{A}_{t} \underline{z_{i|t+1}^{[j]}}+\hat{B}_{t} \underline{v_{i|t+1}^{[j]}}
 \in \underline{\mathbb T^x_{i+1|t+1}} \ominus \underline{\mathbb W_{i,t+1}},
\end{align*}
 $\mathbb S_{t+1}\subseteq \underline {\mathbb S_{t+1}}$, $\mathbb T^x_{i|t+1}\subseteq \underline{\mathbb T^x_{i|t+1}}$, 
$\textbf{co}(\mathbb T^u_{i|t+1})\subseteq \textbf{co}(\underline{\mathbb T^u_{i|t+1}})$,
and $\mathbb W_{i,t+1}\subseteq \underline{\mathbb W_{i,t+1}}$.
Since the prediction dynamics is linear and the input applied to any
$z\in\mathbb T^x_{i|t+1}$ is obtained through the convex-combination
mapping \eqref{convU}, it follows that
\begin{align*}
    \hat{A}_{t+1} z_{i|t+1}^{[j]}+\hat{B}_{t+1} v_{i|t+1}^{[j]}
\in \mathbb T^x_{i+1|t+1}\ominus \mathbb W_{i,t+1}.
\end{align*}
Hence \eqref{cons5m} holds at time $t+1$, and the proposed solution in \eqref{actualsolution} satisfies all the constraints in $\mathbb P2_{t+1}$.
\end{itemize}

Therefore, whenever the backup setup is invoked, a feasible solution to the COCP at time $t+1$ exists. If Criterion~\ref{Crdare2} is satisfied, the algorithm instead attempts to solve the updated problem $\mathbb P2_{t+1}$ corresponding to the modified problem data. If this problem is feasible, the resulting solution is applied; otherwise, the backup setup ensures feasibility. Repeating the same argument at subsequent time steps establishes recursive feasibility of the proposed MPC COCP. \hfill$\blacksquare$
\begin{cor}\label{corolaa}
    Suppose Assumption \ref{Asdare1} holds, and the COCP $\;\mathbb{P}2_t$ is feasible at time $t$. Then, by Lemmas \ref{lemma1}, \ref{lemma2} and Theorem \ref{theo1}, it is guaranteed that the signals $x_t, u_t, \hat \psi_t\in\mathcal L_\infty$.
\end{cor}

% \begin{cor}
%     Suppose Assumption \ref{Asdare1} holds, $x_{max}=\max_{x\in\mathbb X}||x||_Q$ and $u_{max}=\max_{u\in\mathbb U}||u||_R$. Then, by Lemmas \ref{lemma1}, \ref{lemma2} and Theorem \ref{theo1}, the term 
%     \begin{align}\label{gamma1}
%         \gamma_1\triangleq \max_{t\in\mathbb I_0^\infty, x\in\mathbb X, u\in\mathbb U} N(||x||_Q^2+||u||_R^2)+||x||^2_{P_t}
%         \end{align} is finite.
% \end{cor}
\begin{thm}\label{theo2}
    Suppose Assumption \ref{Asdare1} holds, and the COCP $\;\mathbb{P}2_t$ is feasible at time $t$. Then, Algorithm \ref{algoatube} guarantees robust exponential stability of the origin of the plant \eqref{osys1}.
\end{thm}
{\textit{Proof: }}The proof follows the standard MPC stability argument \cite{kouvaritakis2016model} by using the cost function of the COCP as a candidate Lyapunov function. Let the COCP cost function be denoted as
\begin{align*}
    &J_t(x_t,\bar\mu_t)\triangleq\sum_{j=1}^{M_t} \Gamma_t^{[j]},\;\text{where }\\
   & \Gamma_t^{[j]}\triangleq \sum_{i=0}^{N-1} \left(|| z_{i|t}^{[j]}||^2_Q + || v_{i|t}^{[j]}||^2_R \right)+|| z_{N|t}^{[j]}||^2_{P_t}
\end{align*} is the component of the cost function for each vertex of the tube-sections. To address the change in the number of vertices ($M_t$ to $M_{t+1}$) due to a change in the tube geometry, define the variables \begin{align*}
\rho_{1_t}\triangleq\begin{cases}
        0,\;\text{if } M_{t}\leq M_{t+1}\\
         1, \;\text{otherwise,}
    \end{cases}\rho_{2_t}\triangleq\begin{cases}
        0,\;\text{if } M_{t}\geq M_{t+1}\\
         1, \;\text{otherwise}.
    \end{cases}
\end{align*} 
Also, let 
\begin{align}
    \gamma_1\triangleq \max_{t\in\mathbb I_0^\infty, x\in\mathbb X, u\in\mathbb U} N(||x||_Q^2+||u||_R^2)+||x||^2_{P_{t+1}},\label{gamma_1_bound}
\end{align} which by Lemmas \ref{lemma1}, \ref{lemma2} and Theorem \ref{theo1}  is finite.

At any time ${t+1}$, where $t\in\mathbb I_0^\infty$, we can write
\begin{align}
&J_{t+1}^*(x_{t+1})\leq J_{t+1}(x_{t+1},\bar\mu_{t+1})=\sum_{j=1}^{M_{t+1}} \Gamma_{t+1}^{[j]}\nonumber\\
&=\sum_{j=1}^{M_{t}} \Gamma_{t+1}^{[j]}-\rho_{1_t}\sum_{j=M_{t+1}+1}^{M_{t}} \Gamma_{t+1}^{[j]}+\rho_{2_t}\sum_{j=M_t+1}^{M_{t+1}} \Gamma_{t+1}^{[j]}\nonumber\\
&\leq \sum_{j=1}^{M_{t}}\left(\sum_{i=0}^{N-1} \left(|| z_{i|t+1}^{[j]}||^2_Q + || v_{i|t+1}^{[j]}||^2_R \right)+|| z_{N|t+1}^{[j]}||^2_{P_{t+1}}\right) \nonumber\\
&+\rho_{2_t}(M_{t+1}-M_t)\gamma_1 \hspace{3cm}(\text{using \eqref{gamma_1_bound}})\nonumber\\
&\leq \sum_{j=1}^{M_{t}}\left(\sum_{i=0}^{N-1} \left(|| \underline{z_{i|t+1}^{[j]}}||^2_Q + || \underline{v_{i|t+1}^{[j]}}||^2_R \right)+|| z_{N|t+1}^{[j]}||^2_{P_{t+1}}\right) \nonumber\\
&+\rho_{2_t}(M_{t+1}-M_t)\gamma_1 \nonumber\\
&\hspace{1.5cm}\left(\because \mathbb T^x_{i|t+1}\subseteq \underline{\mathbb T^x_{i|t+1}},\;\textbf{co}(\mathbb T^u_{i|t+1})\subseteq \textbf{co}(\underline{\mathbb T^u_{i|t+1}})\right)\nonumber\\
&=\sum_{j=1}^{M_{t}}\left(\sum_{i=0}^{N-2} \left(|| {z_{i+1|t}^{[j]^*}}||^2_Q + || {v_{i+1|t}^{[j]^*}}||^2_R \right)+|| z_{N|t+1}^{[j]}||^2_{P_{t+1}}\right.\nonumber\\
&\left. +||z_{N|t}^{[j]^*}||^2_{Q+K_t^\top R K_t} \vphantom{\sum_{i=1}^{M}}\right) +\rho_{2_t}(M_{t+1}-M_t)\gamma_1 \;\;(\text{using \eqref{nonada}})\nonumber\\
&=J_t^*(x_t)-\sum_{j=1}^{M_t}\left(|| {z_{0|t}^{[j]^*}}||^2_Q + || {v_{0|t}^{[j]^*}}||^2_R + ||{z_{N|t}^{[j]^*}}||^2_{P_t} \right)\nonumber\\
&+\sum_{j=1}^{M_t}\left(|| (\hat A_{t+1}+\hat B_{t+1}K_{t+1})z_{N|t}^{[j]^*} + w_{t+N}||^2_{P_{t+1}}\right.\nonumber\\
&\left. +||z_{N|t}^{[j]^*}||^2_{Q+K_t^\top R K_t}\right)+\rho_{2_t}(M_{t+1}-M_t)\gamma_1 \nonumber\end{align}\begin{align}
&\leq (1-\gamma_2) J_t^*(x_t)+\sum_{j=1}^{M_{t}} \left( -||  z_{N|{t}}^{[j]^*}||^2_{P_t}+||  z_{N|{t}}^{[j]^*}||^2_{Q+K_t^\top R K_t}\right. \nonumber\\
 &\left.+||  z_{N|{t}}^{[j]^*}||^2_{(\hat A_{t+1}+\hat B_{t+1}K_{t+1})^\top P_{t+1}(\hat A_{t+1}+\hat B_{t+1}K_{t+1})} \right)\nonumber\\
 &+M_t \max_{t\in\mathbb I_0^\infty, \;w\in\mathbb W_t}||w ||^2_{P_{t+1}} +\rho_{2_t}(M_{t+1}-M_t)\gamma_1, \nonumber\\&\hspace{3cm}(\text{by Cauchy-Schwartz inequality})\nonumber
\end{align}
where $\sum_{j=1}^{M_t}\left(|| {z_{0|t}^{[j]^*}}||^2_Q + || {v_{0|t}^{[j]^*}}||^2_R\right)=\gamma_2 J_t^*(x_t)$, implying $\gamma_2\in(0,1]\Rightarrow 1-\gamma_2\in[0,1)\;\;\text{(by definition of $J_t^*$)}.$

Since the COCP is initially feasible and Assumption \ref{Asdare1} holds, by Theorem \ref{theo1}, recursive feasibility is guaranteed. Either Criterion \ref{Crdare2} is satisfied, or we revert to $\hat A_{t+1}=\hat A_t$, $\hat B_{t+1}=\hat B_t$, $P_{t+1}=P_t$ and $K_{t+1}=K_t$ that satisfies \eqref{dare1}. In any case, we have 
\begin{align}
   & P_t-(\hat A_{t+1}+\hat B_{t+1}K_{t+1})^\top P_{t+1}(\hat A_{t+1}+\hat B_{t+1}K_{t+1}) \nonumber\\
   &\hspace{4cm} -Q-K_t^\top R K_t\succeq 0, \nonumber\\
   \Rightarrow & J_{t+1}^*(x_{t+1})\leq (1-\gamma_2)J_{t}^*(x_{t})+\rho_{2_t}(M_{t+1}-M_t)\gamma_1  \nonumber\\&\hspace{4cm}+M_t \max_{t\in\mathbb I_0^\infty, \;w\in\mathbb W_t}||w ||^2_{P_{t+1}}  .\nonumber 
\end{align}
For implementation purposes, we put an upper limit $M_{max}$ on the number of vertices $\mathbb S_t$, depending on the computational power. Due to recursive feasibility and bounded sets $\mathbb D$ and $\Psi_t$, the value of $w_t$ is upper bounded. The elements of $P_{t}$ are also upper bounded, since $P_{t}$ is involved in the cost function. Therefore, we can define a constant $\gamma_3\triangleq M_{max} \max_{t\in\mathbb I_0^\infty, \;w\in\mathbb W_t}||w ||^2_{P_{t}}+(M_{max}-M_{min})\gamma_1$ resulting in
\begin{align}
    J_{t+1}^*(x_{t+1})\leq (1-\gamma_2)J_{t}^*(x_{t})+\gamma_3,
\end{align}
which implies the plant \eqref{osys1} exhibits robust exponential stability. Further, once the learning has converged, $M_{t}$ remains constant and the term $M_{t+1}-M_t$ becomes 0, implying that the state eventually converges to a finite set that depends on the sizes of $\mathbb D$ and $\Psi_t$ and $\hat A_t$, $\hat B_t$ as $t\rightarrow\infty $, all of which are bounded (see Corollary \ref{corolaa}).   
\hfill$\blacksquare$

\section{Numerical Example}
We validate the proposed MPC framework using the following $2^{\text{nd}}$ order LTI system\footnote{A $2^{\text{nd}}$ order system is chosen for ease of visualizing the tubes.}.
\begin{align*}
    x_{t+1}=\begin{bmatrix}
        0.2 & 1.015\\-0.2825 & 1
    \end{bmatrix}x_t+\begin{bmatrix}
        1.08\\3
    \end{bmatrix}u_t+d_t
\end{align*}
with $||d_t||_\infty \leq 0.1$ and constraints $||x_t||_\infty \leq 20$, $||u_t||_\infty \leq 10$. The parametric uncertainty is considered in two elements of $A$ and one element of $B$ with $3$ vertices of $\Psi$ given by {$\psi^{[1]}=\begin{bmatrix}
    0.2 & 1.3 & 0.7&;&-1 & 1 & 3
\end{bmatrix}$, $\psi^{[2]}=\begin{bmatrix}
    0.2 & 1.2 & 0.9&;&-0.25 & 1 & 3
\end{bmatrix}$, $\psi^{[3]}=\begin{bmatrix}
    0.2 & 1 & 1.1 &;&0.35 & 1 & 3
\end{bmatrix}$}. The results in Fig.s \ref{xufig}-\ref{RMSJ} are provided for $x_0=[18\;;-18]$, $N=10$, $Q=\begin{bmatrix}
    1 & 0&;&0 &1
\end{bmatrix}$, $R=0.1$ and $\kappa=0.9$ with $\hat \psi_0=(\psi^{[1]}+\psi^{[2]}+\psi^{[3]})/3$.

\begin{figure}[t]
\centering
\framebox{\parbox{3in}{
\includegraphics[scale=0.415]{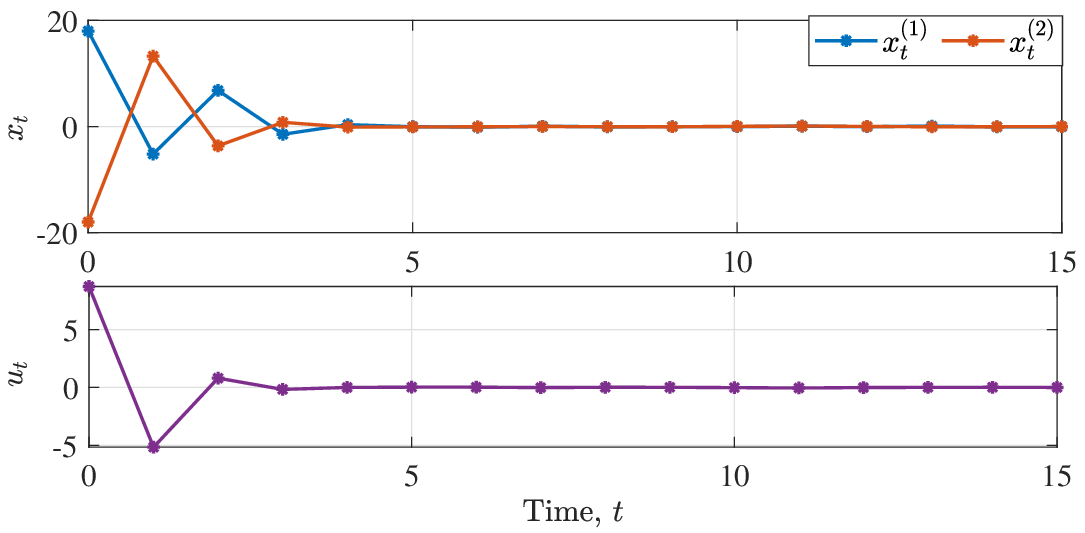}}} 
\caption{State and control input.}  
\label{xufig}   
\end{figure}

\begin{figure}[t]
\centering
\framebox{\parbox{3in}{
\includegraphics[scale=0.417]{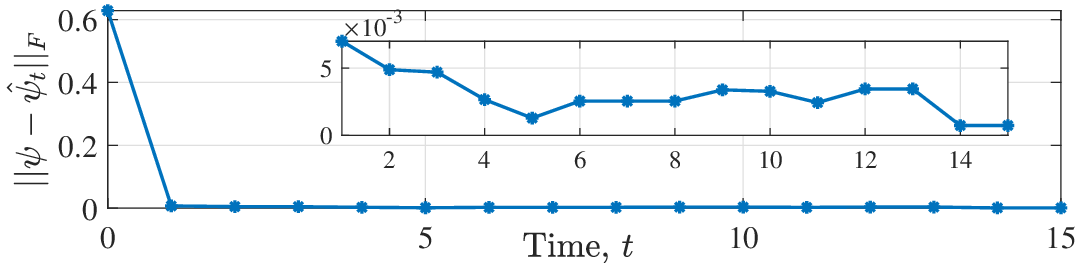}}} 
\caption{Frobenius norm of the parameter estimation error.}  
\label{ABfig}   
\end{figure}
\begin{figure}[t]
\centering
\framebox{\parbox{3in}{
\includegraphics[scale=0.414]{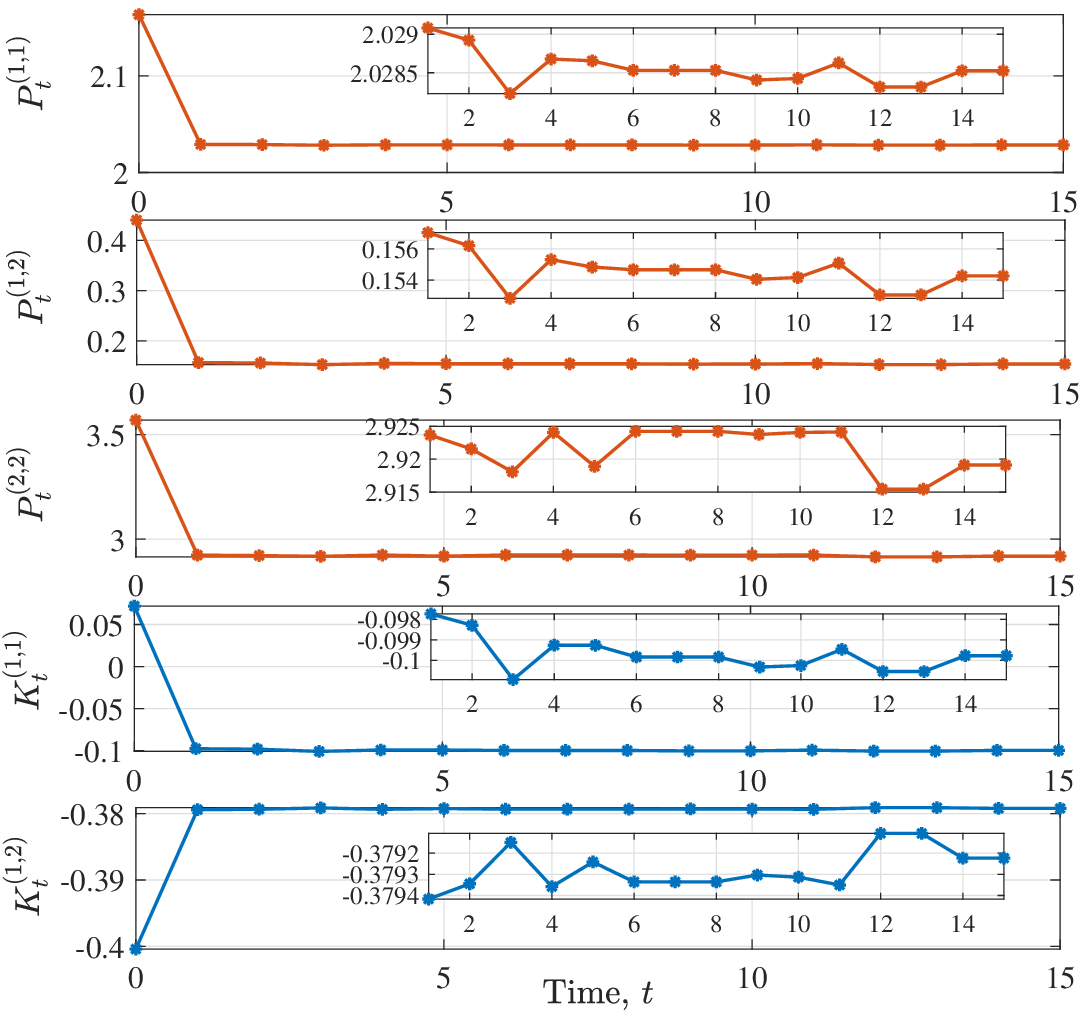}}} 
\caption{Elements of the terminal cost weight matrix, and the stabilizing linear feedback gain.}  
\label{KPfig1}   
\end{figure}

It is seen from Fig. \ref{xufig} that the hard constraints in \eqref{hc1} are satisfied, and the state trajectory converges to the neighbourhood of the origin. The parameter estimation error norm is bounded as shown in Fig. \ref{ABfig}, and the variations in the elements of $P_t$ and $K_t$ are shown in Fig. \ref{KPfig1}. The reduction of the size of the parametric uncertainty set $\Psi_t$ at $t=1,3,11$ is shown in Fig. \ref{psifig}. From Fig. \ref{BPXTSfig} (a), (b) and (c), the adaptation in the shape of the homothetic tube $\mathbb S_t$ and the terminal set $\mathbb X_{TS_t}$ is visible. The shape of $\mathbb S_t$ changes and the size reduces, giving a better characterization of state propagation in the presence of uncertainty $\Psi_t$ and $\mathbb D$, whereas the terminal set size increases thereby reducing the strain on the controller to reach the terminal set.
\begin{figure}[t]
\centering
\framebox{\parbox{3in}{
\includegraphics[scale=0.415]{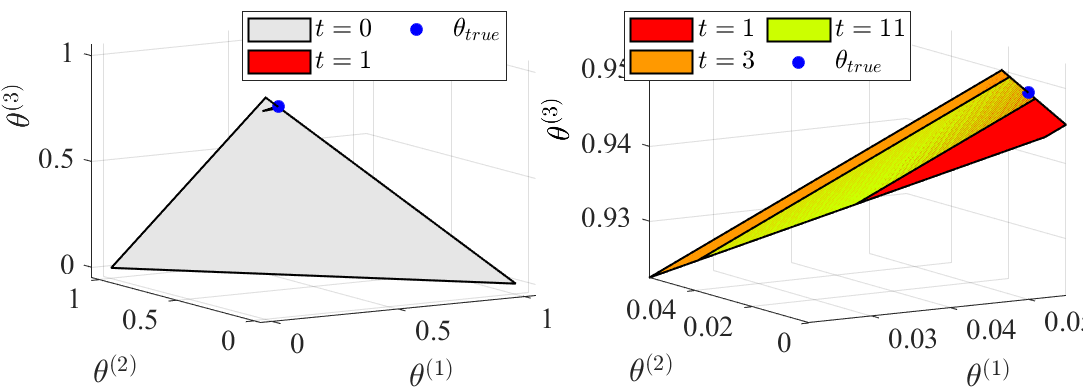} }}
\caption{Barycentric representation of the parametric uncertainty sets containing the convex combination of the true parameter (shown with a blue dot).}  
\label{psifig}     \end{figure}
\begin{figure}[t]
\centering
\framebox{\parbox{3in}{
\includegraphics[scale=0.415]{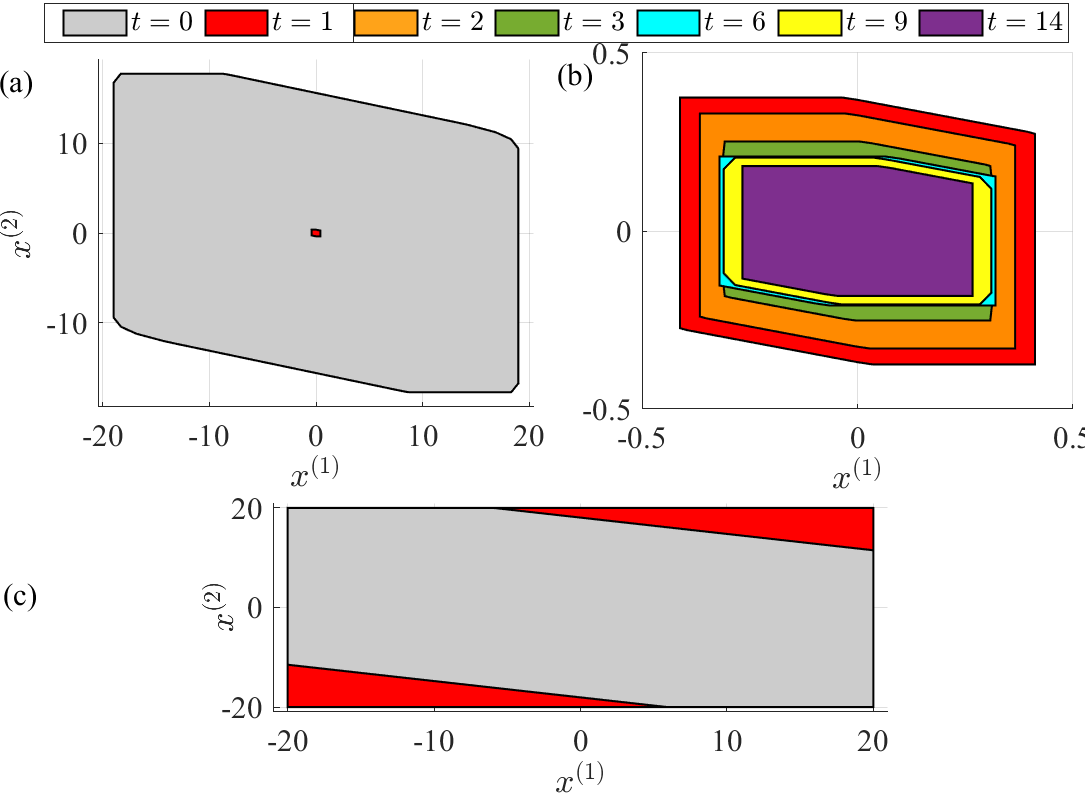}}} 
\caption{(a), (b) The polytopes used to define the cross-sectional shape of the homothetic tubes, (c) The terminal sets.}  
\label{BPXTSfig}   
\end{figure}
To highlight the impact of tube adaptation, the performance of the proposed MPC is compared with the robust homothetic tube MPC scheme of \cite{rakovic2012homothetic}. The comparison focuses on the amount of constraint tightening required to guarantee robustness, the stabilization performance and the associated stage cost.

Fig.~\ref{tubexfig} shows the predicted state tubes at selected time instants. The first and second subplots correspond to the tubes generated at $t=0$ and $t=1$, respectively. The tube produced by the proposed adaptive method (shown in green) is noticeably smaller than the tubes obtained using the robust homothetic tube MPC method of \cite{rakovic2012homothetic} (shown with magenta outline). This reduction in tube size is a direct consequence of the contraction of the parametric uncertainty set through online learning.

The second subplot further shows that the tube generated by the proposed method without adaptation (blue outline) is still smaller than the tube generated following \cite{rakovic2012homothetic}. This is because even without adaptation, the disturbance set characterization is improved at each time step, whereas the robust tube MPC scheme relies on a fixed set that accounts for all admissible disturbance $w_t$, thereby leading to larger tube cross-sections.

The tightening sets associated with the disturbance bounds are illustrated in Fig.s~\ref{fig:ConstraintTigh} (a) and (b) for $t=0$ and $t=1$, respectively. The tightening set used by the method of \cite{rakovic2012homothetic} is shown in grey in Fig.~\ref{fig:ConstraintTigh} (a) and remains fixed over time. In contrast, the tightening sets generated by the proposed method, shown in Fig. \ref{fig:ConstraintTigh} (a) and (b), shrink as the parametric uncertainty set is refined through learning. This reduction enlarges the region available for the state trajectory to evolve while maintaining robust constraint satisfaction.

The improvement in performance is reflected in Fig.~\ref{RMSJ}. The figure shows the evolution of the state norm $\|x_t\|_2$ and the cumulative stage cost for the proposed method and the robust tube MPC scheme of \cite{rakovic2012homothetic}. The proposed adaptive scheme achieves a faster reduction of the state norm, indicating improved stabilization performance. In addition, the cumulative stage cost is lower, reflecting a reduction in control effort enabled by the improved model representation obtained through online learning.

\begin{figure}[t]
\centering
\framebox{\parbox{3in}{
\includegraphics[scale=0.416]{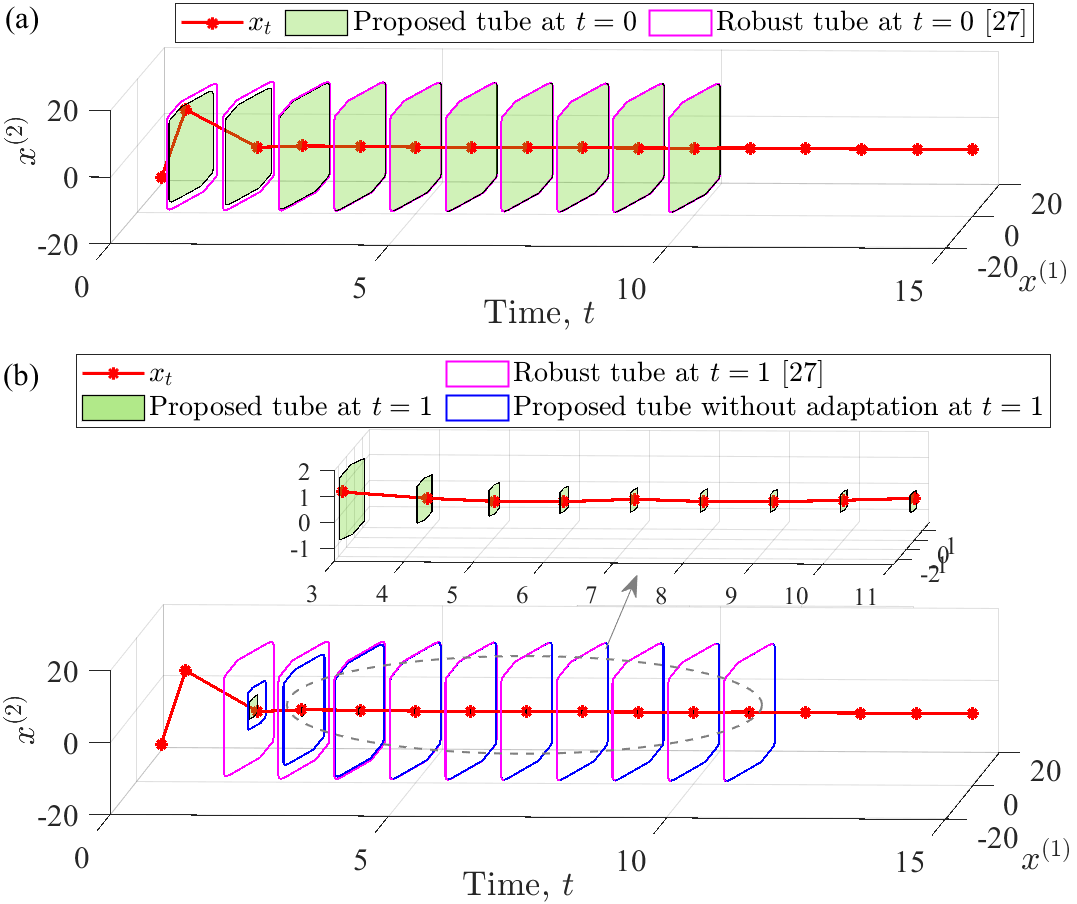}}} 
\caption{Tubes generated at $t=0$ and $t=1$ following the proposed framework and \cite{rakovic2012homothetic}, along with the state trajectory generated with the proposed method.}  
\label{tubexfig}   
\end{figure}
\begin{figure}[t]
\centering
\framebox{\parbox{3in}{
\includegraphics[scale=0.415]{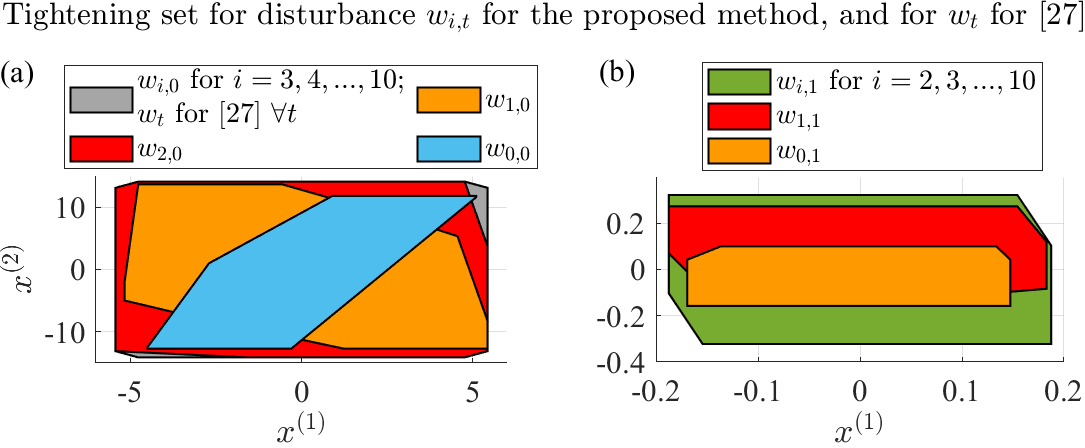}}} 
\caption{Sets used for constraint tightening corresponding to the disturbance bounds for the proposed adaptive method and the robust tube MPC \cite{rakovic2012homothetic}.}  
\label{fig:ConstraintTigh}   
\end{figure}
\begin{figure}[t]
\centering
\framebox{\parbox{3in}{
\includegraphics[scale=0.415]{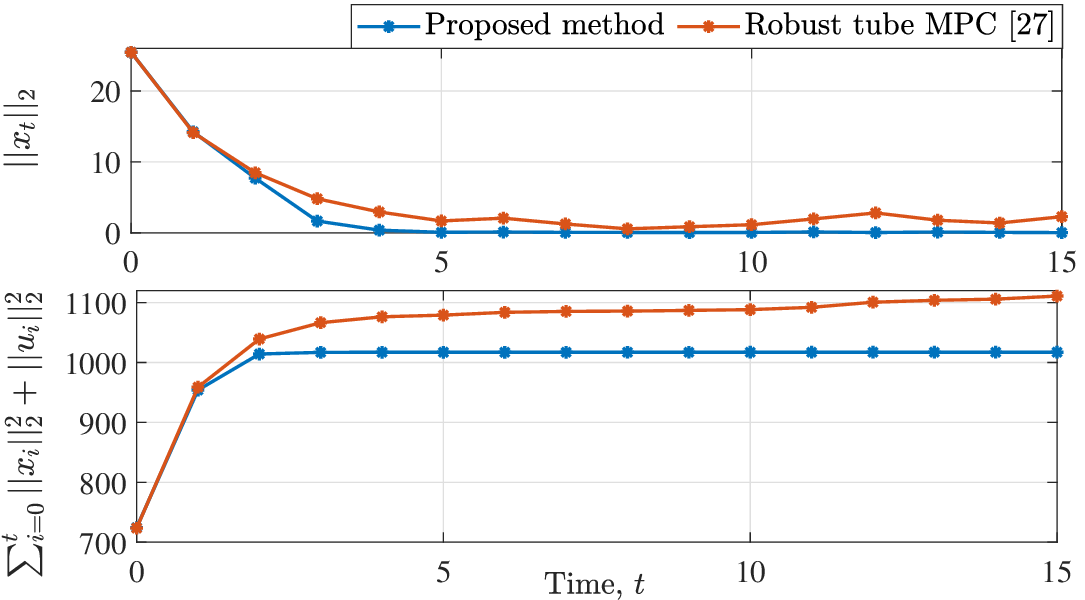}}} 
\caption{Comparison of $||x_t||_2$ and the cumulative stage cost for the proposed method and \cite{rakovic2012homothetic}.}  
\label{RMSJ}   
\end{figure}
These results demonstrate that adapting the tube geometry and other COCP components through online refinement of the uncertainty set significantly reduces conservatism while maintaining robust constraint satisfaction.
\section{Conclusion}
This paper introduces an adaptive tube framework for MPC of discrete-time LTI systems subject to parametric uncertainty and additive disturbances. The approach integrates homothetic tube-based MPC with set-membership identification, allowing the parametric uncertainty set and parameter estimates to be updated online using available state and input data. As the uncertainty set contracts, the tube cross-sections adapt accordingly, leading to less conservative constraint tightening and state propagation. It is formally established that the resulting adaptive tube MPC scheme guarantees recursive feasibility, robust exponential stability, and boundedness of all signals, despite the time-varying nature of the COCP induced by adaptation. A key theoretical feature of the framework is that it avoids the standard assumption of a common quadratically stabilizing linear feedback gain for the entire parametric uncertainty set, instead requiring only a one-step Lyapunov compatibility condition between consecutive parameter updates. A backup mechanism ensures that these guarantees are preserved even when the adaptive updates render the modified COCP infeasible. Numerical results demonstrate the benefits of the adaptive tube construction in reducing conservatism and improving performance. Future work will focus on reducing the computational burden associated with online recomputation of tube and terminal ingredients and extending the framework to the output feedback scenario.

% OR

%\begin{figure}
%\begin{center}
%\epsfig{file=jcaesar,width=7cm}
%\caption{Gaius Julius Caesar, 100--44 B.C.}
%\label{fig1}
%\end{center}
%\end{figure}

% \begin{ack}                               % Place acknowledgements
% .  % here.
% \end{ack}

\bibliographystyle{plain}        
\bibliography{autosam}

% \appendix
% \section{Computing the non-falsified set}\label{apndx1}
% Since $\Psi$ is a convex set, the true parameter can be written as 
% \begin{align}
%     \psi=\sum_{i=1}^{L}\theta^{[i]} \psi^{[i]}, \;\;\sum_{i=1}^{L}\theta^{[i]} =1,\;\theta^{[i]}\in[0,1]
% \end{align} where $\theta\triangleq \begin{bmatrix}
%     \theta^{[1]} & \theta^{[1]}  & ... & \theta^{[L]} 
% \end{bmatrix}^\top\in\Theta$, $\Theta\triangleq\{ \eta\in\mathbb R^L\;\middle|\; \}$ is unique for $\psi$ and $\Psi$. This implies there is a bijection between the vertices of $\Psi$  In fact, for any of the updated sets $\Psi_t$, the vertices
% \section{Some Latin vocabulary}       
                                       
\end{document}